\documentclass[a4paper,11pt,reqno]{amsart}

\usepackage{relsize}
\usepackage{cite}
\usepackage{color}
\usepackage[dvipsnames]{xcolor}

\usepackage{hyperref, enumitem}
\hypersetup{
  colorlinks   = true, 
  urlcolor     = blue, 
  linkcolor    = Purple, 
  citecolor   = red 
}
\usepackage{amsmath,amsthm,amssymb}
\usepackage[margin=2.4cm]{geometry}
\usepackage{stmaryrd}
\usepackage{comment}
\usepackage{xparse}
\usepackage{colortbl}
\usepackage{nicematrix}

\newcommand{\C}{\mathcal{C}}
\newcommand{\A}{\mathcal{A}}

\newcommand{\F}{\mathbb{F}}

\newcommand{\N}{\mathbb{N}}

\newcommand{\T}{\mathcal T}

\newcommand{\m}{\mathfrak{m}}

\newcommand{\Z}{\mathbb{Z}}
\newcommand{\Zr}{\mathbb{Z}_{p^r}}

\DeclareMathOperator{\GR}{GR}

\theoremstyle{definition}
\newtheorem{theorem}{Theorem}[section]
\newtheorem{definition}[theorem]{{{Definition}}}
\newtheorem{example}[theorem]{{{Example}}}
\newtheorem{notation}[theorem]{{{Notation}}}
\newtheorem{remark}[theorem]{{{Remark}}}

\newtheorem{proposition}[theorem]{{{Proposition}}}
\newtheorem{lemma}[theorem]{{{Lemma}}}

\definecolor{light-gray}{gray}{0.90}

\newcommand{\julia}[1]{{\color{blue} \sf $\star\star$ \underline{Julia}: [#1]}}

\usepackage{setspace}

\title[MDP convolutional codes over finite chain rings]{\textbf{Convolutional codes over finite chain rings, MDP codes and their characterization}}

\author[G. N. Alfarano]{Gianira N. Alfarano}
\address{Institute of Mathematics, University of Zurich, Switzerland}
\curraddr{}
\email{gianiranicoletta.alfarano@math.uzh.ch}

\author[A. Gruica]{Anina Gruica}
\address{Department of Mathematics and Computer Science, Eindhoven University of Technology, the Netherlands}
\email{a.gruica@tue.nl}

\author[J. Lieb]{Julia Lieb}
\address{Institute of Mathematics, University of Zurich, Switzerland}
\email{julia.lieb@math.uzh.ch}

\author[J. Rosenthal]{Joachim Rosenthal}
\address{Institute of Mathematics, University of Zurich, Switzerland}
\email{rosenthal@math.uzh.ch}


\date{}

\begin{document}

\maketitle

\begin{abstract}
In this paper, we develop the theory of convolutional codes over finite commutative chain rings. In particular, we focus on maximum distance profile (MDP) convolutional codes and we provide a characterization of these codes, generalizing the one known for fields. Moreover, we relate (reverse) MDP convolutional codes over a finite chain ring  with  (reverse) MDP convolutional codes over its residue field. Finally, we provide a construction of (reverse) MDP convolutional codes over finite chain rings generalizing the notion of (reverse) superregular matrices.
\end{abstract}

\section{Introduction}
A convolutional code is usually defined as an $\F[z]$-submodule of $\F[z]^n$, where $\F[z]$ denotes the polynomial ring over a finite field $\F$. This family of codes has been exhaustively studied in the past decades because of their flexibility of grouping blocks of information in an appropriate way, according to the erasures location.
Contrary to the block codes case, where the minimum distance plays a crucial role, for convolutional codes, the notion of \textbf{$j$-th column distances} allows to study some decoding properties of the codes. In particular, if as many of these column distances as possible reach the known upper bound for their value, then the code is called \textbf{maximum distance profile} (MDP) convolutional code. It was shown in \cite{virtu2012} that over the erasure channel 
MDP convolutional codes have optimal recovery rate for windows of a certain length, depending on the code parameters. Constructions of MDP convolutional codes over finite fields were proposed in several works; see for example \cite{alfarano2020weighted, li17, stronglymds}.

Codes over finite rings have been analyzed since the 1970s, in particular over the ring of integers modulo $m$. Their study was initiated by Blake, who generalized the notions of Hamming codes, Reed-Solomon codes and BCH codes over arbitrary integer residue rings; see \cite{blake1972codes, blake1975codes}. Spiegel continued the study of BCH codes involving group algebras over rings of $p$-adic integers; see~\cite{spiegel1978codes}. 
At the end of the 90's, in \cite{norton2000structure} Norton and S\v{a}l\v{a}gean generalized the structure theorems of Calderbank and Sloane provided in \cite{calderbank1995modular} for
linear and cyclic codes over the ring of integers modulo $p^r$, which we denote by $\Z_{p^r}$, to a finite chain ring. The structure theorems proposed in their work have been used to prove results on the Hamming distance of codes over a finite chain ring in \cite{hammingrings}, where there is a particular focus on the class of cyclic codes. Cyclic codes over finite chain rings were further investigated in \cite{dinh2004cyclic}. 

Convolutional codes over rings were proposed as appropriate codes for phase modulation over the ring of integers modulo $m$; see \cite{baldini1987systematic, baldini1990coded, massey1989convolutional}. When considered over rings, convolutional codes can have properties which are different from  convolutional codes over fields; these differences depend on the structure of the underlying ring. Fagnani and Zampieri were the pioneers in the study of the theory of convolutional codes over the ring $\Z_{p^r}$, in the case when the input sequence space is a free module; see \cite{fagnani2001system}. The problem of constructing convolutional codes over $\Z_{p^r}$ has been also well-studied in the last decade.

In this paper, we are interested in presenting the theory of convolutional codes over finite (commutative) chain rings. In this setting, which is very general and natural, we relate the algebraic properties of chain rings to the known results on convolutional codes over $\Z_{p^r}$, generalizing them in a natural way. Given a finite chain ring $R$, we define convolutional codes as $R[z]$-submodules of $R[z]^n$, not necessarily free. When the submodule is \emph{delay-free}, we are able to provide a characterization of MDP convolutional codes over $R$, generalizing the results provided in \cite[Theorem 2.4]{stronglymds}. 
We also relate MDP convolutional codes defined over $R$ to MDP convolutional codes defined on the residue field $R/\mathfrak{m}$, where $\mathfrak{m}$ is the unique maximal ideal of $R$. This allows us to obtain constructions of (reverse) MDP convolutional codes over $R$ from existing constructions for convolutional codes over finite fields. Moreover, we propose a definition of \emph{superregularity} which extends the known one for matrices defined over fields. Some particular constructions of (reverse) MDP convolutional codes over finite chain rings are then provided and some examples are presented.

This paper is based on the master’s thesis of the second author \cite{Gruica} and in this paper we extend the results which were originally stated there.  
While completing the manuscript, we came across the recent paper \cite{napp2021noncatastrophic}, in which noncatastrophic convolutional codes over $\Z_{p^r}$ are considered and a characterization of MDP convolutional codes is provided.
As we state in Remark \ref{noncatastrophic}, noncatastrophic codes are always delay-free. 
Moreover, noncatastrophic codes over $\Z_{p^r}$ are always free \cite{napp2021noncatastrophic} and in Remark \ref{rem:delayfree-free} we compare the notions free and delay-free and show that none of the two definitions implies the other.
This implies that our results are in two ways more general than the corresponding results in \cite{napp2021noncatastrophic} since we consider delay-free codes over finite (commutative) chain rings, which are a generalization of the integer residue rings $\Z_{p^r}$.

The paper is organized as follows.
In Section \ref{sec:prel}, we present some preliminaries on finite (commutative) chain rings and modules over such rings. In Section \ref{sec:conv}, we establish the notation for the rest of the paper and we introduce the notion of convolutional codes over chain rings. In Section \ref{sec:MDP}, we investigate MDP convolutional codes and provide a characterization for such codes. In Section \ref{sec:constr}, we present constructions for MDP and reverse MDP convolutional codes over chain rings, starting from an MDP convolutional code over a finite field or from a superregular matrix. We also consider the special case of convolutional codes over integer residue rings in more detail.

\section{Preliminaries}\label{sec:prel}
We start by introducing the necessary background for the setting of our main object of interests, namely convolutional codes over finite chain rings. We refer the interested reader to \cite{bini2012finite,nechaev2008finite,mcdonald1974finite} for more details on finite chain rings.

\subsection{Finite chain rings}
For convenience of the reader, we recall the basics of finite chain rings.

\begin{definition}
A \textbf{finite (commutative) chain ring} is a finite (commutative) ring with identity $1$ whose ideals form a chain under inclusion.  In particular, finite chain rings are local rings, i.e. they have a unique maximal ideal.
\end{definition}

In this paper all the considered rings are commutative if not otherwise specified. 

Many different rings are part of the family of chain rings. One of the most important classes of such rings is the one of Galois rings which includes, but is not limited to, finite fields and rings of the form $\Zr$ where $p$ is a prime number and $r\in\mathbb{N}$. 

Let $R$ be a finite chain ring and $\m$ be its maximal ideal. Since $R$ is a principal ideal ring, there exists an element $\gamma\in R$ such that $\m=\langle \gamma \rangle = \gamma R$.
Define $\nu$ to be the nilpotency index of $\gamma$. Then, there are exactly $\nu+1$ ideals in $R$, $R=\langle \gamma^0\rangle$, $\langle \gamma^1\rangle, \dots, \langle \gamma^\nu\rangle=\{0\}$, which form a chain:
$$R=\langle \gamma^0\rangle \supset \langle \gamma^1\rangle \supset \dots\supset  \langle \gamma^\nu\rangle=\{0\}. $$
\begin{definition}
A Galois ring   $\mathrm{GR}(p^r,s)$ is a finite extension of the ring of integers $\Z_{p^r}$ of degree $s$. In particular, it is a finite commutative chain ring of size $p^{rs}$ and characteristic $p^r$, which  is isomorphic to $\Z_{p^r}[z]/(f(z))$, where $f(z)$ is a monic polynomial of degree $s$, irreducible in $\Z_{p^r}[z]$. 
Note that for a Galois ring $\mathrm{GR}(p^r,s)$, the generator of the unique maximal ideal can always be chosen to be $p$ and $r$ is the smallest integer such that $\m^r=\{0\}$. 
\end{definition}

In general, the residue field $R/\m$ of $R$ is finite, hence it has cardinality $q$, where $q=p^h$ is a prime power. We can choose a canonical set of representatives for $R/\m$ as follows.

\begin{definition}
The unique set of representatives $\T$ for $R/\m$ of the form $\{0,1,\xi,\xi^2,\dots, \xi^{q-2}\}$, where $\xi\in R$ is called \textbf{Teichm\"uller set}. 
\end{definition}
Every element $a\in R$ can be uniquely represented as
\begin{equation}\label{eq:uniquedecompTS}
    a=\sum_{i=0}^{\nu-1}t_i\gamma^i, \quad t_i\in\T.
\end{equation}
Due to the uniqueness of \eqref{eq:uniquedecompTS}, it easily follows that the cardinality of $R$ is $q^\nu$.
For Galois rings, this representation is substituted with the $p$-adic expansion and $\langle \xi\rangle$ coincides with the unique cyclic subgroup of $R^\ast$ of order $p^s-1$. For integer residue rings $\mathbb Z_{p^r}$, the Teichmüller set is considered to be $\A_p:=\{0,1,\hdots,p-1\}$.

\begin{example}
Let $R:=\Z_{8}=\{0,1,\dots,7\}$. Its unique maximal ideal is $\m=\{0,2,4,6\} = \langle 2\rangle$ and $R/\m =\{0,1\} =\F_2$. Finally, $R^\ast =\{1,3,5,7\}$, hence its unique cyclic subgroup of order $1$ is $\{1\}$, so $\T = \{0,1\}$, with $\xi=1$. For example, the Teichm\"uller representation of $a=6$ is $a = 0\cdot \gamma^0 + 1\cdot \gamma^1 + 1\cdot \gamma^2$, where $\gamma=2$.
\end{example}

\begin{example}\textnormal{\cite[Example 3]{renner2020low}}
Let $p = 2, s = 3, r = 3$ and construct $R = \GR(8, 3)$. Consider the ring $\Z_8$, and $h(z) := z^3 +6z^2 +5z +7 \in \Z_8[z]$. The canonical projection of the polynomial $h(z)$ over $\F_2[z]$ is $z^3 + z + 1$ which is primitive in $\F_2[z]$. Thus, we have $R \cong \Z_8[z]/(h(z))$. The maximal ideal of $R$ is $\m = (2)$ and thus we can choose $\gamma = 2$. Moreover, if $\xi$ is a root of $h(z)$, then
we also have $R \cong \Z_8[\xi]$, and every element can be represented as $a_0 + a_1\xi + a_2\xi^2$, for
$a_0, a_1, a_2 \in \Z_8$. Observe that the polynomial $h(z)$ divides $z^7-1$ in $\Z_8[z]$ and therefore $\xi$ has order $7$, and the Teichm\"uller set is $\T = \{0, \xi, \xi^2, \dots, \xi^7= 1\}$.
\end{example}

\subsection{Modules over finite chain rings}
For preliminaries concerning modules over finite chain rings we refer to \cite{feng2014communication}.

Let $R$ be a finite chain ring. It is well-known that an $R$-module is isomorphic to a direct product of some ideals of $R$. This structure can be described by a \textbf{shape}.

\begin{definition}
A \textbf{$\nu$-shape} $\bf{M}_\mu=(\mu_1, \dots, \mu_\nu)$ is a sequence of non-decreasing integers, $0\leq \mu_1\leq \dots \leq \mu_\nu$, such that $\mu=\sum_{i=1}^\nu\mu_i$. For convention, we set $\mu_0=0$.
\end{definition}  
For every $\nu$-shape $\bf{M}_\mu$ define a finite $R$-module $R(\bf{M}_\mu)$ to be
\[
R(\bf{M}_\mu):=\underbrace{\langle 1\rangle \times \dots \times \langle 1\rangle }_{\text{$\mu_1$}}\times \underbrace{\langle \gamma \rangle \times \dots \times \langle \gamma\rangle }_{\text{$\mu_2-\mu_1$}}\times\dots\times \underbrace{\langle \gamma^{\nu-1} \rangle \times \dots \times \langle \gamma^{\nu-1}\rangle }_{\text{$\mu_\nu-\mu_{\nu-1}$}}.
\]
The converse is also true as it is stated by the following theorem.
\begin{theorem}\textnormal{\cite[Theorem 2.2]{honold2000linear}}\label{thm:uniqueshape} For any finite $R$-module $M$ over a finite chain ring $R$, there is a unique $\nu$-shape $\bf{M}_\mu$ such that $M \cong R(\bf{M}_\mu)$.
\end{theorem}

 For every matrix $A\in R^{m\times n}$, we define the \textbf{shape} of $A$ as the shape of the rowspace of $A$ (or equivalently as the shape of the column space of $A$, since they are isomorphic as $R$-modules).

\section{Convolutional codes over finite chain rings}\label{sec:conv}
In this section we introduce convolutional codes over finite chain rings. In order to do so, we need to provide some useful background on block codes defined over a finite chain ring. Almost all the results and definitions can be found in \cite{honold2000linear}.

\medskip
 
 \begin{notation}
From now on, $R$ will denote a finite commutative chain ring, $n$ is a positive integer, $\m=\langle \gamma\rangle$ is the unique maximal ideal of $R$ and $\nu$ its nilpotency index. Finally, we let $\T$ be the Teichm\"uller set of $R/ \mathfrak{m}$. Let $R[z]$ be the polynomial ring in the variable $z$ with coefficients in $R$ and $R[z]^n$ be the space of polynomial vectors of length $n$. Observe that there is a canonical isomorphism between $R[z]^n$ and the space of polynomials whose coefficients are vectors with entries in $R$, namely $R^n[z]$. Hence we use one or the other according to what is more convenient.
\end{notation}

\begin{definition}
Let $\left \{v_1(z), \dots, v_k(z)\right\} \subseteq R[z]^n$. Then \begin{align*}
    \sum_{i=1}^k t_i(z) v_i(z), \; t_i(z) \in \mathcal{T}[z] \text{ for } i = 1, \dots, k
\end{align*}
is said to be a \textbf{$\gamma$-linear combination} of $v_1(z), \dots, v_k(z)$. The set of all $\gamma$-linear combinations of $v_1(z), \dots, v_k(z)$ is called the \textbf{$\gamma$-span} of $\left \{v_1(z), \dots, v_k(z)\right\}$ and is denoted by 
$\gamma\textnormal{-}\mathrm{span}(v_1(z), \dots, v_k(z))$. 
\end{definition}

\begin{definition}
An ordered sequence of vectors $(v_1(z), \dots, v_k(z))$ in $R^n[z]$ is said to be a \textbf{$\gamma$-generator sequence} if 
\begin{itemize}
    \item[(i)] $\gamma v_i(z)$ is a $\gamma$-linear combination of $v_{i+1}(z), \dots, v_k(z)$ for any $i \in \{1,2,\dots, k-1\},$
    \item[(ii)] $\gamma v_k(z) = 0$.
\end{itemize}
We say that a $\gamma$-generator sequence is \textbf{$\gamma$-linearly independent} if there is no non-trivial $\gamma$-linear combination of its vectors that is $0$. A $\gamma$-linearly independent $\gamma$-generator sequence is called \textbf{$\gamma$-basis}. The number of elements in a $\gamma$-basis of an $R[z]$-submodule of $R[z]^n$ is the \textbf{$\gamma$-dimension} of a submodule.
\end{definition}

It has been proved in \cite{minimalringencoders} that two $\gamma$-bases of an $R[z]$-submodule of $R[z]^n$ have the same number of elements and therefore the number of elements in a $\gamma$-basis of a fixed submodule is an invariant, meaning that the concept of $\gamma$-dimension of a submodule is well-defined.

\subsection{Linear block codes over finite chain rings}
The notion of linear block coke over $\Z_{p^r}$ was introduced in \cite{hammingrings}.

\begin{definition}
A \textbf{(linear) block code} $\mathcal{C}$ of length $n$ over $R$ is an $R$-submodule of $R^n$. A \textbf{generator matrix} $\Tilde{G} \in R^{\Tilde{k}\times n}$ is a matrix whose rows form a minimal set of generators of the code $\mathcal{C}$ over $R$. If the rows of $\Tilde{G}$ are linearly independent over $R$, i.e. if the map induced by $\Tilde{G}$ is injective, then $\Tilde{G}$ is called an \textbf{encoder} of $\mathcal{C}$. If $\mathcal{C}$ admits an encoder then $\mathcal{C}$ is called a \textbf{free code.}
\end{definition}

\begin{definition}\label{def:Zpmodules}
Assume that the linear block code $\mathcal{C} \subseteq R^n$ has $\gamma$-dimension $k$. The matrix $G \in R^{k \times n}$ whose rows form a $\gamma$-basis of $\mathcal{C}$ is called \textbf{$\gamma$-encoder} of $\mathcal{C}$. The code then consists of all $\gamma$-linear combinations of the rows of $G$. Therefore we can describe the code $\mathcal{C}$ as
$$\mathcal{C} = \lbrace uG \in R^n \mid u \in \T^k \rbrace.$$
\end{definition}

\begin{remark}
For the special case of rings of the form $\mathbb Z_{p^r}$, one has $\gamma=p$. Therefore, we will also use the notions $p$-basis, $p$-encoder, etc. in this setting to be coherent with the notions in the literature on codes over $\mathbb Z_{p^r}$.
\end{remark}

\begin{definition}\cite{hammingrings}
A generator matrix $\Tilde{G}$ of a block code $\C$ is in \textbf{standard form} if after a suitable permutation of the coordinates we can write 

\begin{align}
\Tilde{G} = 
    \begin{bmatrix}
    I_{k_0} & A_{0,1}^0 & A_{0,2}^0 & A_{0,3}^0 & \dots & A_{0,\nu-1}^0& A_{0,\nu}^0 \\
    0 & \gamma I_{k_1} & \gamma A_{1,2}^1 & \gamma A_{1,3}^1 & \dots & \gamma A_{1,\nu-1}^1 & \gamma A_{1,\nu}^1 \\
    0 & 0 & \gamma^2I_{k_2} & \gamma^2A_{2,3}^2 & \dots & \gamma^2A_{2,\nu-1}^2 & \gamma^2A_{2,\nu}^2 \\
    \vdots & \vdots & \vdots & \vdots & \ddots & \vdots & \vdots \\
    0 & 0 & 0 & \dots & 0 & \gamma^{\nu-1}I_{k_{\nu-1}} & \gamma^{\nu-1}A_{\nu-1,\nu}^{\nu-1}
    \end{bmatrix},
\end{align}
where $I_{k_i}$ denotes the identity matrix of size $k_i$. Clearly the columns of a generator matrix $\Tilde{G}\in R^{\Tilde{k} \times n}$ must add up to $n$, so $\Tilde{G}$ is grouped into blocks of $k_0, k_1, \dots, k_{\nu-1}$ and $n- \sum_{i=0}^{\nu-1}k_i$ columns. Observe that $k_i=\mu_{i+1}-\mu_i$, where $\textbf{M}_\mu=(\mu_1,\dots,\mu_{\nu})$ is the $\nu$-shape of the $R$-module generated by $G$, after setting $\mu_0=0$. 
\end{definition}

In~\cite{napp2018column} a standard form for $\gamma$-encoders was introduced in the following way.

\begin{definition}
A $\gamma$-encoder $G$ of a block code $\C$ over $R$ is in \textbf{$\gamma$-standard form} if 
\begin{align} \label{pfrommatrix}
\begin{small}G = 
    \begin{bmatrix}
    I_{k_0} & A_{0,1}^0 & A_{0,2}^0 & A_{0,3}^0 & \dots & A_{0,\nu-1}^0 & A_{0,\nu}^0 \\ 
    ---- & ---- & ---- & ---- & ---- & ---- & ---- \\
    \gamma I_{k_0} & 0 & \gamma A_{1,2}^0 & \gamma A_{1,3}^0 & \dots & \gamma A_{1,\nu-1}^0 & \gamma A_{1,\nu}^0 \\
    0 & \gamma I_{k_1} & \gamma A_{1,2}^1 & \gamma A_{1,3}^1 & \dots & \gamma A_{1,\nu-1}^1 & \gamma A_{1,\nu}^1 \\ 
    ---- & ---- & ---- & ---- & ---- & ---- & ---- \\
    \gamma^2I_{k_0} & 0 & 0 & \gamma^2A_{2,3}^0 & \dots & \gamma^2A_{2,\nu-1}^0 & \gamma^2A_{2,\nu}^0 \\
    0 & \gamma^2I_{k_1} & 0 & \gamma^2A_{2,3}^1 & \dots & \gamma^2A_{2,\nu-1}^1 & \gamma^2A_{2,\nu}^1 \\
    0 & 0 & \gamma^2I_{k_2} & \gamma^2A_{2,3}^2 & \dots & \gamma^2A_{2,\nu-1}^2 & \gamma^2A_{2,\nu}^2 \\
    ---- & ---- & ---- & ---- & ---- & ---- & ---- \\
    \vdots & \vdots & \vdots & \vdots & \hdots & \vdots & \vdots \\ 
    ---- & ---- & ---- & ---- & ---- & ---- & ---- \\
    \gamma^{\nu-1}I_{k_0} & 0 & 0 & 0 & \dots & 0 & \gamma^{\nu-1}A_{\nu-1,\nu}^0 \\
    0 & \gamma^{\nu-1}I_{k_1} & 0 & 0 & \dots & 0 & \gamma^{\nu-1}A_{\nu-1,\nu}^1 \\
    0 & 0 & \gamma^{\nu-1}I_{k_2} & 0 & \dots & 0 & \gamma^{\nu-1}A_{\nu-1,\nu}^2 \\
    0 & 0 & 0 & \gamma^{\nu-1}I_{k_3} & \dots & 0 & \gamma^{\nu-1}A_{\nu-1,\nu}^3 \\
    \vdots & \vdots & \vdots & \vdots & \ddots & \vdots & \vdots \\
    0 & 0 & 0 & 0 & \hdots & \gamma^{\nu-1}I_{k_{\nu-1}} & \gamma^{\nu-1}A_{\nu-1,\nu}^{\nu-1}
    \end{bmatrix}.
    \end{small}
\end{align}
\end{definition}

If $\mathcal{C}$ has $\gamma$-dimension $k$ then clearly $k=\sum_{i=0}^{\nu-1}k_i(\nu-i) = \sum_{i=0}^{\nu-1}(\mu_{i+1}-\mu_i)(\nu-i) $ by counting the rows of the matrix (\ref{pfrommatrix}). 

\begin{definition}
Let $G$ be a $\gamma$-encoder of a linear block code $\mathcal{C}$ over $R$ as in (\ref{pfrommatrix}). The scalars $k_0,k_1,\dots, k_{\nu-1}$ are called the \textbf{parameters} of $\mathcal{C}.$
\end{definition}

It is not difficult to see that any nonzero block code $\mathcal{C}$ over $R$ has a generator matrix in standard form; see also 
\cite[Theorem 3.3]{hammingrings}. Furthermore, all generator matrices of $\mathcal{C}$ have the same parameters $k_0,k_1,\dots, k_{\nu-1}$. This follows immediately from Theorem \ref{thm:uniqueshape} on the uniqueness of the $\nu$-shape of a submodule of $R$.

\begin{remark}
A block code over $R$ is free if and only if its parameters are $k_0 = k$ and $k_i = 0$ for $i = 1, \dots , \nu-1$.
\end{remark}

\begin{definition}
Given an integer $\nu\geq 1$ and an integer $k\geq 0$, we call an ordered vector $(k_0,k_1,\dots,k_{\nu-1}) \in \mathbb{N}^r$ a \textbf{$\nu$-optimal set of parameters of $k$} if $$ k_0 + k_1 + \dots + k_{\nu-1} = \min\left\{\Tilde{k}_0 + \Tilde{k}_1 + \dots + \Tilde{k}_{\nu-1} \mid k = \sum_{i=0}^{\nu-1}\Tilde{k}_i(\nu-i)\right\}.$$
\end{definition}

\begin{remark} \label{uniqueroptimal}
When $\nu$ divides $k$ then $(k_0,0,\dots,0)$ with $k_0 = \frac{k}{\nu}$ is the unique $\nu$-optimal set of parameters of $k$. In general, the $\nu$-optimal set of parameters of $k$ does not have to be unique for a given $k$ and $\nu$. Consider for instance $k = 16$ and $\nu = 5$ then $(3,0,0,0,1)$ and $(0,4,0,0,0)$ are two $5$-optimal sets of parameters of $16$.
\end{remark}

Considering a $\nu$-optimal set of parameters of $k$ as the parameters of a linear block code over $R$ with $\gamma$-dimension $k$, in \cite{napp2017mds}, the following two results were provided for codes defined over $\Z_{p^r}$,  which leads to the following bound on the distance of a linear code over $R$. Here, we state them for codes defined over a finite chain ring $R$.

\begin{lemma} \label{r-optimal}
Let $(k_0,k_1,\dots,k_{\nu-1})$ be a $\nu$-optimal set of parameters of $k$. Then $$k_0 +k_1 + \dots + k_{\nu-1}= \left\lceil \frac{k}{\nu} \right \rceil.$$
\end{lemma}

Using the previous result on the set of parameters of $k$, the following Singleton-like bound for  the distance of linear block codes over $R$ has been derived in terms of the $\gamma$-dimension.

\begin{theorem}\label{singblock}
Let $\mathcal{C} \subseteq R^n$ be a linear code of $\gamma$-dimension $k$. Then
\begin{align} \label{singring}
    d(\mathcal{C}) \leq n- \left\lceil \frac{k}{\nu} \right \rceil + 1.
\end{align}
We call $\mathcal{C}$ an MDS code if the minimum distance attains this bound with equality.
\end{theorem}

\subsection{Convolutional codes over finite chain rings}
It is common to consider convolutional codes over a ring only as $\Z_{p^r}[z]$-modules of $\Z_{p^r}[z]^n$. In the following we generalize some known results for the ring  of integers modulo $p^r$ to the more general setting of finite chain rings. Almost all the proofs are easy to generalize and for the sake of completeness we only propose a few of them which are significant for the rest of the paper or not completely straightforward.

We say that a polynomial vector $v(z)$ in $R^n[z]$, written as $v(z) = \sum_{i=0}^l v_i z^i$, with $v_i \in R^n$ and $v_l \neq 0$, has \textbf{degree} $l$, denoted by $\deg v(z) = l$. We call $v_l$ the \textbf{leading coefficient vector} of $v(z)$ and we denote it by $v^{lc}$. For  $G(z) \in R[z]^{k \times n}$ we denote by $G_{\infty} \in R^{k \times n}$ the matrix whose rows are constituted by the leading coefficients of the rows of $G(z)$.

\begin{definition}
A $\gamma$-basis $(v_1(z), \dots, v_k(z))$ is called a \textbf{reduced $\gamma$-basis} if the vectors $v_1^{lc}, \dots , v_k^{lc}$ are $\gamma$-linearly independent in $R^n$.
\end{definition}

In \cite[Theorem 3.12]{kuijper2007predictable} when $R=\Z_{p^r}$ it is shown that every submodule of $R^n[z]$ has a reduced $\gamma$-basis. The proof for a finite commutative chain ring works in the same way as its original form.

\begin{definition}
A convolutional code $\mathcal{C}$ of length $n$ over $R$ is an $R[z]$-submodule of $R^n[z]$. A \textbf{generator matrix} $\Tilde{G}(z) \in R^{\Tilde{k}\times n}[z]$ of $\mathcal{C}$ is a polynomial matrix whose rows form a minimal set of generators of $\mathcal{C}$ over $R[z]$ and therefore
\begin{align*}
    \mathcal{C} = \lbrace u(z)\Tilde{G}(z) \in R^n[z] \mid u(z) \in R^{\Tilde{k}}[z] \rbrace.
\end{align*}
If $\Tilde{G}(z)$ has linearly independent rows, i.e. if the map induced by $\Tilde{G}(z)$ is injective, then it is called an \textbf{encoder} of $\mathcal{C}$ and $\mathcal{C}$ is a free code. 
\end{definition}

Since not every convolutional code over $R$ admits an encoder due to the fact that a submodule of $R^n[z]$ is not necessarily free, we give the definition of the $\gamma$-encoder as we did for linear block codes.

\begin{definition}\label{def:codeasRmmodule}
A \textbf{$\gamma$-encoder} $G(z) \in R[z]^{k\times n}$ of a convolutional code $\mathcal{C}$ is a polynomial matrix whose rows form a $\gamma$-basis of $\mathcal{C}$. Then
\begin{align*}
    \mathcal{C} = \lbrace u(z)G(z) \in R^n[z] \mid u(z) \in \mathcal{T}[z]^k \rbrace.
\end{align*}
\end{definition}

If the rows of $G(z)$ form a reduced $\gamma$-basis then we say that $G(z)$ is in \textbf{reduced form} and we call $G(z)$ a \textbf{minimal} $\gamma$-encoder. The row degrees of any minimal $\gamma$-encoder are invariants of the code $\mathcal{C}$, i.e. they are the same for every minimal $\gamma$-encoder of $\C$; see \cite{kuijper2007predictable} for the proof of this result for modules over $\Z_{p^r}$. The sum of these row degrees is the \textbf{$\gamma$-degree} of $\mathcal{C}$. 

We denote an $R[z]$-submodule $\mathcal{C}$ of $R^n[z]$ with $\gamma$-dimension $k$ and $\gamma$-degree $\delta$ an $(n,k,\delta)$-convolutional code.

\begin{definition}
Let $v(z) = \sum_{i=0}^tv_iz^i \in R^n[z]$ with $v_i \in R^n$. The \textbf{weight} of $v(z)$ is defined as $$\operatorname{wt}(v(z)) = \sum_{i=0}^{t}\operatorname{wt_H}(v_i),$$
where $\operatorname{wt_H}(v_i)$ denotes the usual Hamming weight of $v_i$, i.e., the number of nonzero entries in $v_i$.
The \textbf{free distance} of a convolutional code $\mathcal{C}$ is defined as $$ \mathrm{d}_{\operatorname{free}}(\mathcal{C}) = \min\{ \operatorname{wt}(v(z)) \mid v(z) \in \mathcal{C}, \; v(z) \neq 0 \}.$$
\end{definition}
 
Observe that, since $R[z]^n\cong R^n[z]$, if $G(z)\in R[z]^{k\times n}$, then it can also be seen as $G(z)=\sum_{i=0}^m G_iz^i \in R^{k\times n}[z]$, which sometimes is a more convenient form for a $\gamma$-encoder of a convolutional code.

\begin{definition}
Let $\C$ be an $(n,k,\delta)$ convolutional code defined over $R$ and $G(z)=\sum_{i=0}^m G_iz^i\in R^{k\times n}[z]$ be a $\gamma$-encoder for $\C$. For any $j\in\N_0$, we define the \textbf{$j$-th truncated sliding generator matrix} as
\begin{align*}
& G_j^c :=\begin{bmatrix}
G_0 & G_1 & \cdots  & G_j \\
 & G_0 & \cdots & G_{j-1} \\
 & & \ddots & \vdots \\
 & &  & G_0\\
\end{bmatrix}\in R^{(j+1)k\times (j+1)n},\\
\end{align*}
where $G_j = 0,$ whenever $j>m$.
\end{definition}

The following Singleton-like bound for convolutional codes over finite chain rings can be proven in exactly the same way as the corresponding bound for convolutional codes over $\mathbb Z_{p^r}$ in \cite{mdsoverring}. Adapting its proof to finite chain rings is not completely straightforward, so for convenience of the reader we propose it in this setting.

\begin{theorem}\cite[Theorem 4.10]{mdsoverring}
The free distance of an $(n,k,\delta)$-convolutional code $\mathcal{C}$ over $R$ satisfies
\begin{align} \label{gensingring}
    \mathrm{d}_{\operatorname{free}}(\mathcal{C}) \leq n \left( \left\lfloor \frac{\delta}{k}\right\rfloor + 1 \right) - \left\lceil \frac{k}{\nu} \left( \left\lfloor \frac{\delta}{k} \right\rfloor  + 1 \right) - \frac{\delta}{\nu} \right\rceil + 1.
\end{align}
\end{theorem}

\begin{proof}
Let $G(z)\in R[z]^{n\times k}$ be a $\gamma$-encoder of the convolutional code and define $$\C_j:=\{v(z)=u(z)G(z)\ |\  u(z)\in\T^k[z], \deg(v(z))\leq j \}.$$
$\C_j$ is a submodule of the module of polynomial vectors with degree at most $j$, denoted by $R[z]^n_{\leq j}$, which is isomorphic to $R^{n(j+1)}$. Hence we can view $\C_j$ as a linear block code, whose $\gamma$-dimension we denote by $k_j^{\gamma}$. Using Theorem \ref{singblock}, one obtains
\begin{align*}
    \mathrm{d}_{\operatorname{free}}(\C)&\leq\min_{j\geq 0}\max_{k_j^{\gamma}}(d(\C_j))\leq \min_{j\geq 0}\max_{k_j^{\gamma}}\left(n(j+1)-\left\lceil\frac{k_j^{\gamma}}{\nu}\right\rceil+1\right)\\
    &=\min_{j\geq 0}\left(n(j+1)-\left\lceil\frac{\min(k_j^{\gamma})}{\nu}\right\rceil+1\right)\leq \min_{j\geq \lfloor\frac{\delta}{k}\rfloor}\left(n(j+1)-\left\lceil\frac{\min(k_j^{\gamma})}{\nu}\right\rceil+1\right).
\end{align*}
The reason why it is helpful to consider the last inequality is \cite[Corollary 4.8]{mdsoverring}, which can be directly translated to the more general setting of finite chain rings, and states that $$k_j^{\gamma}\leq\max((j+1)k-\delta),0)=(j+1)k-\delta\quad \text{for}\quad j\geq \left\lfloor\frac{\delta}{k}\right\rfloor.$$ It follows
$$\mathrm{d}_{\operatorname{free}}(\mathcal{C})\leq\min_{j\geq \lfloor\frac{\delta}{k}\rfloor}\left(n(j+1)-\left\lceil\frac{(j+1)k-\delta}{\nu}\right\rceil+1\right) = n \left( \left\lfloor \frac{\delta}{k}\right\rfloor + 1 \right) - \left\lceil \frac{k}{\nu} \left( \left\lfloor \frac{\delta}{k} \right\rfloor  + 1 \right) - \frac{\delta}{\nu} \right\rceil + 1.$$
\end{proof}

\begin{proposition} \cite[Proposition 14]{napp2018column} \label{pgenseq}
If $G(z) \in R[z]^{k\times n}$ is a $\gamma$-encoder of a convolutional code, then the rows of the $j$-th truncated sliding generator matrix form a $\gamma$-generator sequence, for any $j\in \mathbb{N}_0$.
\end{proposition}

Observe that even if the rows of $G_j^c$ form a $\gamma$-generator sequence for a convolutional code over $R$, they do not always form a $\gamma$-basis, since they are not necessarily $\gamma$-linearly independent. 

The following is an example of a $2$-encoder of a convolutional code over $\mathbb{Z}_4$, whose first truncated sliding generator matrix does not have $\gamma$-linearly independent rows.

\begin{example}
Consider $\gamma=2$ and the $2$-encoder 
\begin{align*}
    G(z) = \begin{bmatrix}
    1 + z & 1 + z & 1 + z \\
    2 + 2z & 2 + 2z & 2 + 2z \\
    0 & 0 & 2z^2
    \end{bmatrix} \in \mathbb{Z}_4[z]^{3 \times 3}.
\end{align*}
Then the rows of 
\begin{align*}
    G_1^c = \begin{bmatrix}
    G_0 & G_1\\
    & G_0
    \end{bmatrix}
    = \begin{bmatrix}
    1 & 1 & 1 & 1 & 1 & 1 \\
    2 & 2 & 2 & 2 & 2 & 2 \\
    0 & 0 & 0 & 0 & 0 & 0 \\
    &  &  & 1 & 1 & 1 \\
    &  &  & 2 & 2 & 2 \\
    & & & 0 & 0 & 0
    \end{bmatrix} \in \mathbb{Z}_4^{6 \times 6}
\end{align*}
are not $2$-linearly independent.
\end{example}

We define a subclass of convolutional codes over $R$ with the property that  the rows of the $j$-th truncated sliding generator matrix form a $\gamma$-basis.

\begin{definition}
A $\gamma$-encoder $G(z)$ of a convolutional code $\mathcal{C} \subseteq R^n[z]$ is said to be \textbf{delay-free} if $G(0) = G_0$ has $\gamma$-linearly independent rows.
\end{definition}

It was shown in \cite[Lemma 3.31]{toste2016distance} (for $\mathbb Z_{p^r}$ but the generalization is straightforward)
that if an $(n,k,\delta)$-convolutional code $\mathcal{C}$ admits a delay-free $\gamma$-encoder, then all the $\gamma$-encoders of $\mathcal{C}$ are delay-free. Hence, for convenience, from now on we call a convolutional code that admits a delay-free $\gamma$-encoder a \textbf{delay-free convolutional code}.

\begin{remark}\label{rem:delayfree-free}
There exist convolutional codes that are free but not delay-free and convolutional codes that are delay-free but not free. We will give examples for both ways in the following.\\
Firstly, $\tilde{G}(z)=[ z\ z]\in\mathbb Z_{p^r}[z]^2$ is the generator matrix of a free code that is not delay free as $$G(z)=\left[\begin{array}{cc}
z     &  z\\
pz     & pz\\
\vdots & \vdots\\
p^{r-1}z & p^{r-1}z
\end{array}\right]$$
is a corresponding $p$-encoder with $G(0)=0$.

Secondly, take $A(z)\in\mathcal A_p[z]^{k\times n}$ with linearly independent rows and such that also $A(0)$ is full rank. Define $\tilde{G}(z)=p^{r-1}A(z)$. Then obviouly, the code with generator matrix $\tilde{G}(z)$ is not free as $p\cdot \tilde{G}(z)=0$, i.e. the rows of $\tilde{G}(z)$ are not linearly independent over $\mathbb Z_{p^r}[z]$. However, the rows of $\tilde{G}(z)$ are $p$-linearly independent (as the rows of $A(z)$ are linearly independent) and hence $G(z)=\tilde{G}(z)$ is a $p$-encoder and $G(0)=p^{r-1}A(0)$ has $p$-linearly independent rows as $A(0)$ is full rank. Hence the code is delay-free.
\end{remark}

\section{MDP convolutional codes over finite chain rings}\label{sec:MDP}

In this section we focus on the class of MDP convolutional codes over finite chain rings. After generalizing some results from $\Z_{p^r}$ to $R$, we can prove the main theorem of the paper, which is a characterization of MDP convolutional codes over finite chain rings.

If $G(z)$ is a delay-free $\gamma$-encoder, then the rows of $G_j^c$ are $\gamma$-linearly independent for any $j \in \mathbb{N}$, because the rows of $G_0$ are $\gamma$-linearly independent.

\begin{definition}
The \textbf{$j$-th column distance} of a convolutional code $\mathcal{C}$ is defined for any $j>0$ as 
$$d_j^c(\C) := \min \{\operatorname{wt_H}(v_0+v_1z+\cdots+v_jz^j) \mid v(z) = \sum_iv_iz^i \in \mathcal{C}, \, v_0 \neq 0 \}.$$
\end{definition}

It clearly holds that $d_0^c \leq d_1^c \leq d_2^c \dots \text{ and } \lim\limits_{j \rightarrow \infty}{d_j^c} = \mathrm{d}_{\mathrm{free}}(\mathcal{C})$.

For a delay-free $(n,k,\delta)$-convolutional code $\mathcal{C}$, the $j$-th column distances can be defined using the $j$-th truncated sliding generator matrix $G_j^c$ of $\mathcal{C}$ as
\begin{align} \label{cd}
    d_j^c = \min \{\operatorname{wt_H}((u_0 \; \dots \; u_j)G_j^c) \mid u_i \in \mathcal{T}^k, \, u_0 \neq 0 \},
\end{align}
which is usually a more convenient definition to use. For delay-free convolutional codes the $j$-th column distance defined as in (\ref{cd}) is well-defined, i.e. the $j$-th column distance does not depend on the choice of the generator matrix $G(z)$.
If we have a delay-free convolutional code $\mathcal{C}$ we can write $G_0$ in $\gamma$-standard form as in (\ref{pfrommatrix}) where $k_0,k_1, \dots, k_{\nu-1}$ are the fixed parameters of the linear block code with $\gamma$-encoder $G_0$, derived from the shape of the submodule of $R^n$ generated by $G_0$.

\begin{theorem} \cite[Theorem 3.35]{toste2016distance} \label{columndistancefirstbound}
Let $\mathcal{C}$ be a delay-free $(n,k,\delta)$-convolutional code in $R^n[z]$ with $\gamma$-encoder $G(z)$ such that the block code with $\gamma$-encoder $G_0$ has parameters $k_0,k_1, \dots, k_{\nu-1}$. Then 
\begin{align*}
    d_j^c \leq \begin{cases}
    (j+1)\left( n- \sum_{i=0}^{\nu-j}k_i \right) - \sum_{s=2}^j sk_{\nu-(s-1)} +1   & \text{ for } j \leq \nu, \\
    (j+1)n - \sum_{i=0}^{\nu-1}k_i -k - (j-\nu)k_0 +1  &\text{ for } j > \nu.
    \end{cases}  
\end{align*}
\end{theorem}

From Theorem \ref{columndistancefirstbound} it follows that - in order for the bound on the column distances of the code to be maximal - we would like to choose the $\nu$-optimal set of parameters of $k$, such that the value of $k_0$ is as great as possible. Indeed, if we look at the two bounds on the $j$-th column distances in Theorem \ref{columndistancefirstbound}, we can easily see that out of all the parameters $k_0, k_1, \dots, k_{\nu-1}$, the parameter which least affects the bound is $k_0$. Therefore we choose $k_0, \dots,k_{\nu-1}$ to be
\begin{align} \label{choosek_0}
    k_0 = \left\lfloor \frac{k}{\nu} \right\rfloor, \; k_{\nu-N} = 1 \text{ and } k_i = 0,
\end{align}
where $N = k - \left\lfloor \frac{k}{\nu} \right\rfloor \nu$ and $i = 1, \dots, \nu-1, \; i \neq \nu-N$. Indeed, this forms a $\nu$-optimal set of parameters of $k$: 
\begin{align*}
\nu k_0 + (\nu-1)k_1 + \dots + k_{\nu-1} = \nu\left\lfloor \frac{k}{\nu} \right\rfloor + N = \nu\left\lfloor \frac{k}{\nu} \right\rfloor +  k - \left\lfloor \frac{k}{\nu} \right\rfloor \nu = k.
\end{align*}
With this $\nu$-optimal set of parameters we can maximize the bound in Theorem \ref{columndistancefirstbound}.

\begin{theorem} \cite[Corollary 3.36]{toste2016distance} \label{establishedcd}
The $j$-th column distance of a delay-free $(n,k,\delta)$-convolutional code in $R^n[z]$ satisfies
\begin{align}
    d_j^c \leq \begin{cases}
    \left( n- \left\lceil \frac{k}{\nu} \right\rceil \right) (j+1) +1  & \textnormal{ for } j \leq N \\
     \left( n- \left\lceil \frac{k}{\nu} \right\rceil \right) (j+1) - \left(\left\lceil \frac{k}{\nu} \right\rceil - \left\lfloor \frac{k}{\nu} \right\rfloor \right)(N+1)+1  & \textnormal{ for }j > N,
    \end{cases} 
\end{align}
where $N = k - \left\lfloor \frac{k}{\nu} \right\rfloor \nu$.
\end{theorem}

For a delay-free $(n,k,\delta)$-convolutional code $\mathcal{C}$ over $R$ all of the column distances are smaller or equal than the free distance. In particular, they must also be smaller or equal than the value established in the generalized Singleton bound for convolutional codes over $R$ in~\eqref{gensingring}. Let $L$ be the largest integer such that the bound on the $L$-th column distance is smaller or equal to the generalized Singleton bound.
It is possible to explicitly determine a formula for $L$ (see \cite{napp2018column} for the case $R=\Z_{p^r}$). Since for some cases these formulas are rather long and complicated and in this paper we only need an explicit formula for the case that $\nu\mid k$, we only give the expression for $L$ for this case.

\begin{lemma} \label{L}
Let $\mathcal{C}$ be an $(n,k,\delta)$-convolutional code over $R$ with $\nu\mid k$. Then 
\begin{align*}
    L = \left\lfloor \frac{\delta}{k} \right\rfloor + \left\lfloor \frac{\left\lfloor \frac{\delta}{\nu} \right\rfloor}{n -  \frac{k}{\nu}}\right\rfloor.
\end{align*}
\end{lemma}

At this point we are finally able to provide the object of interest of this work.

\begin{definition}
A delay-free $(n,k,\delta)$-convolutional code $\mathcal{C}$ over $R$ is of \textbf{maximum distance profile} (MDP) if the $j$-th column distances attain the bounds established in Theorem \ref{establishedcd} for all $j \leq L$.
\end{definition}

\begin{lemma}\cite[Lemma 3.41]{toste2016distance}\label{ifjtheni}
Let $\mathcal{C}$ be a delay-free  $(n,k,\delta)$-convolutional code over $R$ with $\nu|k$, such that the parameters of the linear block code with $\gamma$-encoder $G_0$ are $k_0=\frac{k}{\nu}$ and $k_i=0$ for $i=1,\dots,\nu-1$. If
\begin{align*}
d_j^c = \left( n- \frac{k}{\nu} \right) (j+1) +1 \text{ for some } j \in \mathbb{N}, 
\end{align*}
then 
\begin{align*}
d_i^c = \left( n-\frac{k}{\nu} \right) (i+1) +1 \text{ for all } i < j.
\end{align*}
\end{lemma}

The following theorem is the main result of this section. It provides a characterization of MDP convolutional codes over finite chain rings and it is inspired by \cite[Theorem 2.4]{stronglymds}.

\begin{theorem} \label{general}
Let $G(z) \in R^{k \times n}[z]$ be a $\gamma$-encoder of a delay-free convolutional code $\mathcal{C} \subseteq R^n[z]$ with $\nu|k$, and such that the parameters of the linear block code with $\gamma$-encoder $G_0$ are $k_0=\frac{k}{\nu}$ and $k_i=0$ for $i=1,\dots,\nu-1$. Let $G_j^c$ be the $j$-th truncated sliding generator matrix of $\mathcal{C}$. Then the following statements are equivalent.
\begin{itemize}
    \item[(a)] $d_j^c(\mathcal{C}) = \left(n-\frac{k}{\nu}\right)(j+1)+1$
    \vspace{0.2cm}
    \item[(b)] Every $(j+1)k \times (j+1)\frac{k}{\nu}$ submatrix of $G_j^c$ formed by selecting the columns with indices $1 \leq t_1 < \dots < t_{(j+1)\frac{k}{\nu}}$, where $t_{s\cdot\frac{k}{\nu}+1} > sn$ for $s = 1, 2, \dots j$, has $\gamma$-linearly independent rows.
\end{itemize}
In particular, setting $j=L$, this gives an algebraic characterization of MDP convolutional codes over $R$.
\end{theorem}

\begin{proof}
(a) $\implies$ (b): We show that if (b) does not hold then also (a) does not hold.  Assume there are indices $1 \leq t_1 < \dots < t_{(j+1)\frac{k}{\nu}}$ where $t_{s\cdot\frac{k}{\nu}+1}>sn$ for $s= 1,\dots , j$ and that there exists $0 \neq u = (u_0 \; \dots \; u_j) \in \mathcal{T}^{(j+1)k}$ such that $uG_j^c$ has zero coordinates at positions $t_1,\dots,t_{(j+1)\frac{k}{\nu}}$. \\
Let $\ell := \min\{i\mid u_i \neq 0\}$.
If $\ell=0$, then $uG_j^c\in \mathcal{C}$ and \begin{align*}
\operatorname{wt}(uG_j^c) < \left(n-\frac{k}{\nu}\right)(j+1) +1.
\end{align*}
If $\ell>0$ consider the vector $(u_{\ell}, \dots, u_j)G_{j-\ell}^c \in R^{(j-\ell+1)n}$. Since $uG_j^c$ has at least $(j+1)\frac{k}{\nu}$ zero entries and we have $t_{\ell \cdot\frac{k}{\nu}+1}>\ell n$ by assumption, $(u_{\ell} \; \dots \; u_j)G_{j-\ell}^c$ must have at least $(j+1)\frac{k}{\nu} - \ell \frac{k}{\nu} = (j-\ell+1)\frac{k}{\nu}$ zeros. Therefore the weight of $(u_{\ell} \; \dots \; u_j)G_{j-\ell}^c$ is at most $n(j-\ell+1)- (j-\ell+1)\frac{k}{\nu} = (n-\frac{k}{\nu})(j-\ell+1)$ and then it follows $d_{j-\ell}^c \leq (n-\frac{k}{\nu})(j-\ell+1)$. By  Lemma \ref{ifjtheni} it follows that $d_j^c \leq (n-\frac{k}{\nu})(j+1)$.
\\
(b) $\implies$ (a): Assume that (a) does not hold. Then let $m$ be defined as 
\begin{align*}
m := \min\left\{j \mid d_j^c(\C) < \left(n-\frac{k}{\nu}\right)(m+1) + 1\right\}.
\end{align*}
It follows that there is a nonzero vector $u = (u_0, u_1, \dots, u_m) \in \mathcal{T}^{(m+1)k}$ such that $uG_m^c$ has at least $\frac{k}{\nu}(m+1)$ zeros. By assumption we have $k_0 = \frac{k}{\nu}$ in the $\gamma$-standard form of $G_0$ and therefore we can write
\begin{align*}
    G_0 =  \begin{bmatrix}
    I_{k_0} & A_{0} \\ 
    \gamma I_{k_0} & A_{1} \\
    \gamma^2I_{k_0} & A_{2} \\
    \vdots & \vdots\\ 
    \gamma^{r-1}I_{k_0} & A_{\nu-1}\\
    \end{bmatrix}
=   \begin{bmatrix}
   G_0^0 & G_0^1
\end{bmatrix}
\end{align*}
for some $A_i \in R^{k_0 \times (n-k_0)}$ for $i =1,2,\dots, \nu-1$ and where $G_0^0 \in R^{k \times k_0}$ and $G_0^1 \in R^{k \times (n-k_0)}$. Note that the rows of $G_{0}^0$ form a $\gamma$-basis of $R^{k_0}$. Since $G_i \in R^{k\times n}$ we can always write $G_i = \begin{bmatrix}
   G_i^0 & G_i^1
\end{bmatrix}$
where $G_i^0 \in R^{k \times k_0}$ and $G_i^1 \in R^{k \times (n-k_0)}$.
As a submatrix inside $G_j^c$, select the columns corresponding to the first $(m+1)k_0$ positions where $uG_m^c$ has a zero entry and add the $(j-m)k_0$ columns of $G_j^c$ highlighted in Eq. \eqref{eq:bluepart}.
\begin{equation}\label{eq:bluepart}
    G_j^c = 
\left[\begin{array}{{c|>{\columncolor{cyan!20}}cc>{\columncolor{cyan!20}}ccc>{\columncolor{cyan!20}}cc}}
     & {G_{m+1}^0} & G_{m+1}^1 & {G_{m+2}^0} & G_{m+2}^1 & \dots  & {G_{j}^0} & G_{j}^1 \\
     
     & {G_{m}^0} & G_{m}^1 & {G_{m+1}^0} & G_{m+1}^1  & \dots  & {G_{j-1}^0} & G_{j-1}^1 \\
     
     G_m^c  & {G_{m-1}^0} &G_{m-1}^1 & {\vdots} & \vdots & \dots & {\vdots} &G_{j-2}^1 \\
     
     & {\vdots} & \vdots  &  {\vdots} & \vdots & \dots & {\vdots} & \vdots \\
    
     & {G_{1}^0} &G_{1}^1 & {\vdots} & \vdots & \dots & {G_{j-m}^0} &G_{j-m}^1 \\ \\
    
    \hline \\
    
     & {G_{0}^0} & G_{0}^1  &   {G_{1}^0} & G_{1}^1 &  \dots  & {G_{j-m-1}^0}&G_{j-m-1}^1 \\
     
     & & & {G_{0}^0} & G_{0}^1  & \dots& {G_{j-m-2}^0} &G_{j-m-2}^1 \\
     
     & & & & & \ddots & {\vdots} & \vdots \\

     & & & & & & {G_{0}^0}  & G_{0}^1
  \end{array}\right].
\end{equation}

Let $t_1,t_2, \dots , t_{(j+1)k_0}$ be the indices of the selected columns.
The submatrix obtained above is then a $(j+1)k \times (j+1)k_0$ matrix and the indices $t_1,t_2, \dots , t_{(m+1)k_0}$ satisfy $t_{sk_0+1}>sn$ for $s=1, \dots, j$ and that this submatrix has $\gamma$-linearly dependent rows. For proving this, note that $d_i^c = (n-k_0)(i+1) + 1$ for $i=0,1, \dots, m-1$ since we chose $m$ to be the smallest integer which does not satisfy this equality. This implies that $(u_0 \; u_1 \; \dots \; u_i)G_i^c$ has at most $n(i+1)-((n-k_0)(i+1)+1) = k_0(i+1)-1$ zeros for $i=0,1,\dots, m-1$. Therefore, we get that $t_{s\cdot k_0+1}>sn$ for $s=1,\dots ,m$.

Let $M$ be the selected submatrix. It remains to show that there exist a nonzero vectors $\bar{u}=(u_0,\dots, u_j)\in\T^{(j+1)k}$ such that $uM=0$. 
First of all, note that $M$ can be written as follows:
\begin{align*}
    M= \begin{bmatrix}
    X & Y_1 & Y_2 & Y_3 & \dots & Y_{j-m}\\
    &  G_0^0 & G_1^0 & G_2^0 & \dots & G_{j-m-1}^0 \\
    & & G_0^0 & G_1^0 & \dots & G_{j-m-2}^0\\
    & & & \ddots & \ddots & \vdots\\
    & & &  & G_0^0 & G_1^0\\
    & & &  & & G_0^0\\
    \end{bmatrix},
\end{align*}
where $X \in R^{(m+1)k \times (m+1)k_0}$ is the submatrix of $G_m^c$ whose left kernel is nonzero and $Y_i \in R^{(m+1)k \times k_0}$ for $i = 1 , \dots, j-m$. Let $u = (u_0, u_1, \dots, u_m) \in \mathcal{T}^{(m+1)k}$ such that $uX = 0$. Then $uY_1 \in R^{k_0}$ and since the rows of $G_{0}^0$ form a $\gamma$-basis of $R^{k_0}$, we can choose $u_{m+1} \in \mathcal{T}^k$ such that $uY_1 + u_{m+1}G_0^0 = 0$.

Next we choose $u_{m+2} \in \mathcal{T}^k$ such that $uY_2 + u_{m+1}G_1^0 + u_{m+2}G_0^0 = 0$.

Iterating this procedure, we can construct $\bar{u}=(u_0, \dots, u_j) \in \mathcal{T}^{(j+1)k}$ such that $\bar{u}M=0$, which concludes the proof.
\end{proof}

\subsection{Reverse MDP convolutional codes}
In this short subsection, we  introduce reverse 
MDP convolutional codes, which are a generalization of MDP convolutional codes.

\begin{definition}\cite{hutchinson2009existence}
Let $\mathcal{C}$ be an $(n,k,\delta)$ convolutional code with minimal $\gamma$-encoder $G(z)$, which has entries $g_{ij}(z)$ and column degrees $\delta_1,\hdots,\delta_k$. Set $\overline{g_{ij}(z)}:=z^{\delta_j}g_{ij}(z^{-1})$. Then, the code $\overline{\mathcal{C}}$ with $\gamma$-encoder $\overline{G(z)}$ having $\overline{g_{ij}(z)}$ as entries, is also an $(n,k,\delta)$ convolutional code, which is called the \textbf{reverse code} to $\mathcal{C}$. 
\end{definition}

\begin{remark}
It is not difficult to see that $v_0+\cdots+v_dz^d\in\overline{\mathcal{C}}$ if and only if $v_d+\cdots+v_0z^d\in\mathcal{C}$; see for instance \cite{hutchinson2009existence}.
\end{remark}

\begin{definition}
Let $\mathcal{C}$ be an MDP convolutional code. If $\overline{\mathcal{C}}$ is also MDP, $\mathcal{C}$ is called \textbf{reverse MDP} convolutional code.
\end{definition}

In \cite{virtu2012}, reverse MDP convolutional codes were characterized in terms of parity-check matrices. As we do not assume that our codes are non-catastrophic, we cannot work with parity-check matrices. Therefore, in the following proposition, we provide a characterization similar to that from \cite{virtu2012} but using the generator matrix of the code.

\begin{proposition}\label{reverse}
Let $\mathcal{C}$ be an $(n,k,\delta)$ MDP convolutional code with $k\mid\delta$. Furthermore, let $G(z) = G_0 + \cdots +G_{\mu}z^{\mu}$ be a minimal $\gamma$-encoder of $\mathcal{C}$ having all row degrees equal to $\frac{\delta}{k}$. Then the reverse
code $\overline{\mathcal{C}}$ has $\gamma$-encoder $\overline{G}(z) = G_{\mu} +\cdots +G_0z^{\mu}$. Therefore, $\mathcal{C}$ is reverse MDP if and only if part $(b)$ of Theorem \ref{general} is true for $G_L^c$ as well as for the matrix
$$\overline{G}^c_L:=\left[\begin{array}{ccc} G_{\mu} & \cdots & G_{\mu-L}\\  & \ddots & \vdots \\ 0 &  & G_{\mu} \end{array}\right].$$

\end{proposition}

\begin{remark}
For finite fields, a further generalization of reverse MDP convolutional codes called complete MDP convolutional codes has been considered \cite{virtu2012}.
To be able to study complete MDP convolutional codes, it would be necessary to restrict ourselves to non-catastrophic codes as complete MDP convolutional codes possess a parity-check matrix by definition. However, this would restrict ourselves to free codes, for which it has already been shown that the characterization of MDP convolutional codes is analogue to finite fields \cite{napp2021noncatastrophic}. Therefore, we do not consider this special class of codes in this paper.
\end{remark}


\section{Construction of (reverse) MDP convolutional codes over finite chain rings}\label{sec:constr}
In this section we propose a construction of (reverse) MDP convolutional codes over finite chain rings. We will treat in a separate section the special case of integer residue rings. In the last part of the section we will introduce the notion of \emph{(reverse) $\gamma$-superregular matrices} as a generalization of superregular ones, in order to construct (reverse) MDP convolutional codes.
\subsection{The general case of finite chain rings}

The following easy proposition provides a starting point for the construction of MDP convolutional codes over a finite chain ring $R$.

\begin{proposition} \label{plinindep}
Let $A \in R^{n\times n}$ and for $i=0,\hdots,\nu-1$, let $A_i$ be a submatrix of $A$ with $n_i$ rows such that $n_{\nu-1}\geq n_{\nu-2}\geq\cdots\geq n_0\geq 1$. If the rows of $A$ are linearly independent in $R$, then 
\begin{align*}
\Tilde{A} = 
    \begin{bmatrix}
    A_0 \\
    \gamma A_1 \\
    \vdots \\
    \gamma^{\nu-1}A_{\nu-1} 
    \end{bmatrix}
\end{align*}
has $\gamma$-linearly independent rows.
\end{proposition}

\begin{proof}
Since the rows of $A$ are linearly independent in $R$, there is no $u = (u_1, \dots, u_{n}) \in R^{n}\setminus\{0\}$ such that $uA = 0$. Suppose there exists $a = (a_0 \; a_1 \; \dots \; a_{\nu-1}) \in \T^{\nu n}$ such that $(a_0 \; \dots \; a_{\nu-1})\Tilde{A}=0$.
This implies $\tilde{a}_0A + \tilde{a_1}\gamma A + \dots + \tilde{a}_{\nu-1}\gamma^{\nu-1}A = 0$,
where $\tilde{a}_i=\begin{pmatrix}a_i\ 0\end{pmatrix}\in\T^{n}$ for $i=0,\hdots,\nu-1$.

Simplifying the left side of this equation we get
$(\tilde{a}_0 + \tilde{a}_1\gamma + \dots + \tilde{a}_{\nu-1}\gamma^{\nu-1})A = 0$.
Since $\tilde{a}_0 + \tilde{a}_1\gamma + \dots + \tilde{a}_{\nu-1}\gamma^{\nu-1} \in R^n$, the above equation cannot hold unless $a_0 = a_1 = \dots = a_{\nu-1} = 0$.  Therefore the rows of $\Tilde{A}$ are $\gamma$-linearly independent.
\end{proof}

\begin{remark}\label{noncatastrophic}
In \cite{napp2021noncatastrophic} noncatastrophic convolutional codes over $\mathbb Z_{p^r}$ are studied.Noncatastrophicity implies that for an encoder $\tilde{G}(z)$, one has that $\tilde{G}(0)$ is full rank, which implies that $\tilde{G}(z)$ is full rank, i.e. the corresponding code is free. This means that the rows of $$G(z)=\begin{pmatrix} \tilde{G}(z)\\p\tilde{G}(z)\\ \vdots\\ p^{r-1}\tilde{G}(z)
\end{pmatrix}$$ form a $p$-basis, i.e. $G(z)$ is a $p$-encoder. Moreover, $\tilde{G}(0)$ full rank implies that also $G(0)$ is full rank and hence, in summary every noncatastrophic convolutional code over $\mathbb Z_{p^r}$ is delay-free.
\end{remark}

\begin{theorem}\label{main}
Let $\nu\mid k$ and $\nu\mid\delta$ and let $\tilde{G}(z)=\sum_{i=0}^{m}\tilde{G}_iz^i$ with $\tilde{G}_i\in(R/\m)^{\frac{k}{\nu}\times n}$ reduced with sum of its row degrees equal to $\frac{\delta}{\nu}$. Then, the following two statements are equivalent:
\begin{enumerate}
    \item The $(n,\frac{k}{\nu},\frac{\delta}{\nu})$ convolutional code $\tilde{\C}$ over the finite field $R/\m$ with generator matrix $\tilde{G}(z)=\sum_{i=0}^{\mu}\tilde{G}_iz^i$ is MDP.
    \vspace{0.2cm}
    \item The $(n,k,\delta)$ convolutional code $\C$ over the finite chain ring $R$ with $\gamma$-encoder $G(z)=\sum_{i=0}^{\mu}G_iz^i$ with 
    $$G_i=\left[\begin{matrix}\tilde{G}_i\\ \vdots\\ \gamma^{\nu-1}\tilde{G}_i\end{matrix}\right]$$
    is MDP, where here by abuse of notation we denote by $\tilde{G}_i$ its lift to $R^{\frac{k}{\nu}\times n}$.
\end{enumerate}
\end{theorem}

\begin{proof}
Assume (1) is true.
Since $\tilde{G}$ is reduced and the sum of its row degrees is $\frac{\delta}{\nu}$, $G$ is $\gamma$-reduced and the sum of its row degrees equals the $\gamma$-degree of $\mathcal{C}$, which is then equal to $\delta$.
Moreover, every $(L+1)\frac{k}{\nu} \times (L+1)\frac{k}{\nu}$ full size submatrix of
\begin{align}\label{blockstilde}
\Tilde{G}_L^c= \begin{bmatrix}
 \Tilde{G}_0 & \Tilde{G}_1 & \cdots & \Tilde{G}_{L} \\
 & \Tilde{G}_0 & \cdots & \Tilde{G}_{L-1} \\
 &  & \ddots &  \vdots \\
 &  &  &  \Tilde{G}_0 \\
\end{bmatrix} \in (R/\m)^{(L+1)\frac{k}{\nu} \times (L+1)n},
\end{align}
where $\Tilde{G}_i \in R^{\frac{k}{\nu} \times n}$ for $i = 0, \dots, L=\left\lfloor \frac{\delta}{k} \right\rfloor + \left\lfloor \frac{\frac{\delta}{\nu} }{n -  \frac{k}{\nu}}\right\rfloor$ formed from columns with indices $1 \leq t_1 < \dots < t_{(L+1)\frac{k}{\nu}}$, where $t_{s\frac{k}{\nu}+1} > sn$ for $s = 1, 2, \dots L$, is full rank.
Hence, by Proposition \ref{plinindep}, every $(L+1)k \times (L+1)\frac{k}{\nu}$ submatrix of
\begin{align} \label{blocks}
G_L^c= \begin{bmatrix}
 G_0 & G_1 & \cdots & G_{L} \\
 & G_0 & \cdots & G_{L-1} \\
 &  & \ddots &  \vdots \\
 &  &  &  G_0 \\
\end{bmatrix} \in R^{(L+1)k \times (L+1)n},
\end{align}
formed from columns with indices $1 \leq t_1 < \dots < t_{(L+1)\frac{k}{\nu}}$, where $t_{s\frac{k}{\nu}+1} > sn$ for $s = 1, 2, \dots L$, has $\gamma$-linearly independent rows.
As $G(z)$ is a $\gamma$-encoder of a delay-free $(n,k,\delta)$ convolutional code by construction (\eqref{blocks} implies that $G_0$ has $\gamma$-linearly independent rows), one can apply Theorem \ref{general} and obtains that $\mathcal{C}$ is MDP. 

Assume now that (2) is true. From the structure of $G(z)$ it follows that $\tilde{G}(z)$ is the generator matrix of an $(n,\frac{k}{\nu},\frac{\delta}{\nu})$ convolutional code. Moreover, \eqref{blocks} implies \eqref{blockstilde} as well as that $\tilde{G}_0$ has full rank. Consequently, $\tilde{C}$ is MDP.
\end{proof}

\begin{remark}
If $k\mid\delta$, Proposition \ref{reverse} implies that the Theorem \ref{main} is also true for reverse MDP convolutional codes.
\end{remark}

Theorem \ref{main} implies that for any (commutative) finite chain ring $R$ and for any $(n,k,\delta)$ with $\nu\mid k$ and $\nu\mid\delta$, one can construct an $(n,k,\delta)$ MDP convolutional code over $R$ from an $(n,\frac{k}{\nu},\frac{\delta}{\nu})$ MDP convolutional code over the finite field $R/\m$.

\begin{example}
In \cite{alfarano2020weighted} a general construction of MDP convolutional codes has been presented. Here we consider the Example 5.2 in \cite{alfarano2020weighted} of an MDP code over $\F_{{11}^5}$. Let $f(z)\in\Z[z]$ be a monic polynomial of degree $5$, such that $f(z) \mod 11$ is irreducible and let $\alpha$ be a root of $f(z)$. Let $\C$ be the convolutional code with parameters $n=7, k=2$ and $\delta=4$ with generator matrix $\Tilde{G}(z) = \Tilde{G_0} + \Tilde{G_1}z+\Tilde{G_2}z^2$, where 

\begin{align*}
\Tilde{G_0} & = 
\begin{pmatrix}
1 & 2 & 3 & 4 & 5 & 6 & 7 \\
1 & 1 & 1 & 1 & 1 & 1 & 1 \\
\end{pmatrix},\\
\Tilde{G_1} & =
\begin{pmatrix}
\alpha  &  8 \alpha &  5\alpha &  9\alpha &  4\alpha &  7\alpha &  2\alpha \\
1 & 4 & 9 & 5 & 3 & 3 & 5
\end{pmatrix},\\
\Tilde{G_2} & =
\begin{pmatrix}
\alpha^4 &  10\alpha^4 & \alpha^4 & \alpha^4  & \alpha^4  &   10 \alpha^4 &   10 \alpha^4 \\
\alpha^2  &  5 \alpha^2 &  4 \alpha^2  & 3\alpha^2 & 9\alpha^2 & 9\alpha^2 &  3\alpha^2
\end{pmatrix}.
\end{align*}

The  $L$-th truncated sliding matrix $G_L^c$ is
\begin{align*}
\Tilde{G_2}^c & =
\begin{pmatrix}
\Tilde{G_0} & \Tilde{G_1} & \Tilde{G_2} \\
    & \Tilde{G_0} & \Tilde{G_1}  \\
    &     & \Tilde{G_0}  \\
    \end{pmatrix}\in \F_{11}[\alpha]^{6\times 21}.\\
\end{align*}
The code $\C$ is MDP over $\F_{11^5}$. Now choose any $r>1$ such that $k/r=2$ and $\delta/r=4$ and consider the Galois ring $R=\mathrm{GR}(11^r,5)\cong \Z[z]/(11^r,f(z))$. Then, the residue field of $R$ is isomorphic to $\F_{11^5}$. Then the matrix $G(z) = G_0+G_1z+G_2z^2$ where 
$$G_i=\begin{bmatrix}
\tilde{G_i}\\
11 \cdot \tilde{G_i}\\
\vdots\\
11^{r-1}\cdot\tilde{G_i}
\end{bmatrix} $$ is the generator matrix of a $(7,2r,4r)$ MDP convolutional code over $R$.
\end{example}

\subsection{The special case of integer residue rings}

In this section, we want to explain why we need the construction idea that we presented in Theorem \ref{main} for the special case that $R=\mathbb Z_{p^r}$. Moreover, at the end of this section, we provide concrete constructions for this special case.

\begin{remark} \label{rem:embed}
There are different ways to embed $\mathbb F_p$ into $\mathbb Z_{p^r}$. One way is to use the fact that $\mathbb F_p$ is isomorphic to $p^{r-1}\mathbb Z_{p^r}\subset\mathbb Z_{p^r}$. Alternatively, one can identify $\mathbb F_p$ (as a set) with $\mathcal{A}_p$ (the Teichmüller set $\T$ for $\mathbb Z_{p^r}$). 
No matter which one of the two is used, one obtains 
$$L_{\mathbb Z_{p^r}}=\left\lfloor\frac{\delta}{k}\right\rfloor+\left\lfloor\frac{\left\lfloor\frac{\delta}{r}\right\rfloor}{n-\frac{k}{r}}\right\rfloor\leq\left\lfloor\frac{\delta}{k}\right\rfloor+\left\lfloor\frac{\delta}{n-k}\right\rfloor = L_{\mathbb F_p}$$
with equality if and only if $\delta<n-k$.
\end{remark}

In the following, we explain why the first way of embedding $\mathbb F_p$ into $\mathbb Z_{p^r}$ described in Remark~\ref{rem:embed} does not work for our aims. If we take the first option, i.e. multiply each entry of the generator matrix with $p^{r-1}$, we obtain $p^{r-1}\tilde{G}(z)\in\mathbb Z_{p^r}[z]^{n\times k}$. If we multiply any of the entries of this matrix by $p$, we obtain zero and hence, the $p$-span of its rows equals the span of its rows and consequently, the $p$-span of $p^{r-1}\tilde{G}(z)$ is a $\mathbb Z_{p^r}[z]$-submodule of $\mathbb Z_{p^r}[z]^{n}$, i.e. a convolutional code (alternatively, one can see this because the rows of $p^{r-1}\tilde{G}(z)$ form a $p$-generator sequence). Moreover, the rows of $p^{r-1}\tilde{G}(z)$ form a $p$-basis of this submodule (as the rows of $\tilde{G}(z)$ are linearly independent in $\mathbb F_p$), i.e. $p^{r-1}\tilde{G}(z)$ is a $p$-encoder of the corresponding convolutional code.
However, this convolutional code over $\mathbb Z_{p^r}$ cannot be MDP. To see this, take the last $\frac{k}{r}(L_{\mathbb Z_{p^r}}+1)$ columns of $p^{r-1}\tilde{G}_{L_{\mathbb Z_{p^r}}}^c$ and consider a $p$-linear combination of the corresponding parts of the rows of $p^{r-1}\tilde{G}_{L_{\mathbb Z_{p^r}}}^c$ of the form $$a_1(z)p^{r-1}\tilde{g}_1(z)+\cdots+a_{k(L_{\mathbb Z_{p^r}}+1)}(z)p^{r-1}\tilde{g}_{k(L_{\mathbb Z_{p^r}}+1)}(z)$$ with $a_i(z)\in\mathcal{A}_p(z)$ and $\tilde{g}_i(z)\in\mathcal{A}_p(z)^{1\times \frac{k}{r}(L_{\mathbb Z_{p^r}}+1)}$. This linear combination is equal to zero if and only if $a_1(z)\tilde{g}_1(z)+\cdots+a_{k(L_{\mathbb Z_{p^r}}+1)}(z)\tilde{g}_{k(L_{\mathbb Z_{p^r}}+1)}(z)$ is a multiple of $p$ and hence the projection to $\mathbb F_p[z]$ is zero. As the $\tilde{g}_i(z)$ are $k(L_{\mathbb Z_{p^r}}+1)$ vectors of length $\frac{k}{r}(L_{\mathbb Z_{p^r}}+1)$, (for $r>1$) they cannot be linearly independent in $\mathbb F_p[z]$, and hence the corresponding rows in $p^{r-1}\mathbb Z_{p^r}$ cannot be $p$-linearly independent. In follows with Theorem \ref{general} that the convolutional code with $p$-encoder $p^{r-1}\tilde{G}(z)$ cannot be MDP.

Now consider the second option to lift a generator matrix of an MDP code over $\mathbb F_p[z]$ to $\mathbb Z_{p^r}[z]$. In this case, the rows of the generator matrix generate a $\mathbb Z_{p^r}[z]$-submodule of $\mathbb Z_{p^r}[z]$, i.e. a convolutional code, but they form no $p$-generator sequence. This is true because the last row of a $p$-generator sequence has to be a multiple of $p$, i.e. would be zero in $\mathbb F_p$, which contradicts the MDP property.

The construction of Theorem \ref{main} overcomes these problems by viewing the generator matrix of an MDP convolutional code over $\mathbb F_p$ as a polynomial matrix over $\mathbb Z_{p^r}$ and constructing a $p$-encoder out of it.

\medskip 
To construct (reverse) MDP convolutional codes over $\mathbb Z_{p^r}$ with the aid of Theorem \ref{main}, one first needs a construction for (reverse) MDP convolutional codes over prime fields. The following theorem presents - to the best of our knowledge - the only general (i.e. for all code parameters) construction for (reverse) MDP convolutional codes over prime fields using the generator matrix of the code. 
In \cite{li17} this theorem is written in terms of the parity-check matrix and we translate it to the case of generator matrices here.

\begin{theorem}\cite{li17}
  Let $n,k,\delta\in\mathbb N$ with $k<n$ and $k\mid\delta$ and $m=\frac{\delta}{k}$ as well as $p\in\mathbb P$.\\ Then, $G(z)=\sum_{i=0}^{m}G_iz^i\in\mathbb F_p^{k\times n}[z]$ with

$$G_0=\left[\begin{array}{ccccc}
              \binom{m n+n-k}{n-k} & \hdots & 1 &  & 0 \\
              \vdots &   &  & \ddots &  \\
              \binom{m n+n-k}{n-1}  &  & \hdots  &  & 1
            \end{array}\right],$$

$$\hspace{3cm} G_i=\left[\begin{array}{ccc}
              \binom{m n+n-k}{(i+1)n-k} & \hdots & \binom{m n+n-k}{in-k+1} \\
              \vdots &  & \vdots \\
              \binom{m n+n-k}{(i+1)n-1} & \hdots & \binom{m n+n-k}{in}
            \end{array} \right] \mbox{ for } i=1,\hdots,m-1,$$

$$G_{m}=\left[\begin{array}{ccccc}
                  1 &  & \hdots &  & \binom{m n+n-k}{n-1} \\
                    & \ddots &  &  & \vdots \\
                  0 &  & 1 & \hdots & \binom{m n+n-k}{n-k}
                \end{array}\right]$$
              is the generator matrix of an $(n,k,\delta)$ reverse MDP convolutional code if
              $$p>\binom{m n+n-k}{\lfloor (m n+n-k)/2\rfloor}^{k(L+1)}\cdot(k(L+1))^{k(L+1)/2},$$
              where $L=\frac{\delta}{k}+\lfloor\frac{\delta}{n-k}\rfloor$.
            \end{theorem}

\begin{remark}
We can observe the following:
\begin{itemize}
    \item [(i)] From Theorem \ref{main}, we know that in order to construct an $(n,k,\delta)$ (reverse) MDP convolutional code over the ring $\mathbb Z_{p^r}$, we need to construct an $(n,\frac{k}{r},\frac{\delta}{r})$ (reverse) MDP convolutional code over the finite field $\mathbb F_p$. However, for the sake of readability, in the previous theorem, we wrote down the construction over $\mathbb F_p$ for a code with parameters $(n,k,\delta)$.
    \item[(ii)] The lower bound for $p$ in the previous theorem is far from being strict and for concrete examples, one can usually work with much smaller field sizes as illustrated in the following example.
\end{itemize}
\end{remark}

\begin{example}
For the construction of an $(n,\frac{k}{r},\frac{\delta}{r})$ reverse MDP convolutional code over $\mathbb F_p$ with parameters $n=3$, $k=2$, $\delta=2$ and $r=2$, one has $m=1$ and $L=1$. Therefore, in the preceding theorem, the lower bound for $p$ is equal to $200$. However, with this construction, one obtains $G_0=[10\ 5\ 1]$ and $G_1=[1\ 5\ 10]$. One can easily check that the corresponding code is already reverse MDP over $\mathbb F_7$ (identifying $10$ with $3$).
\end{example}

\subsection{Superregular matrices}
In this last section we introduce the  notion of (reverse) $\gamma$-superregular matrices over a finite chain ring $R$, which is the natural generalization of superregular matrices. Superregular matrices with elements in a finite field appear in the context of coding theory, especially for the construction of MDP convolutional codes; we refer to \cite{al16,stronglymds} for a formal definition and details on the relation between superregular matrices and MDP convolutional codes.

Also in this case, we consider $R$ to be a finite chain ring, whose unique maximal ideal is $\mathfrak{m}=\langle \gamma \rangle$, with nilpotency index $\nu$. 
\begin{definition}
Let $\ell$ be a positive integer. A matrix $T$ of the form
\begin{align*}
T=(a_{ij})=\begin{bmatrix}
a_1 & a_{2} & \cdots & a_{{\ell}} \\
a_{-2} & a_1 & \ddots & \vdots \\
\vdots & \ddots & \ddots &  a_{2} \\
a_{-\ell} & \cdots & a_{-2} &  a_1 \\
\end{bmatrix} \in R^{\ell \times \ell}
\end{align*}
is called a \textbf{Toeplitz matrix}.
\end{definition}

Let $A \in R^{k\times n}$. Let $I={i_1,\dots,i_{\ell}}$ and $J={j_1,\dots,j_s}$. We denote by $A^I_J$ the $\ell \times s$ submatrix of $A$ obtained from $A$ by picking the rows with indices $i_1,\dots,i_\ell$ and the columns with indices $j_1,\dots,j_s$.

\begin{definition}
Consider an upper triangular Toeplitz matrix
\begin{align*}
A = 
\begin{bmatrix}
a_1 & a_2 & \cdots & a_{\ell} \\
0 & a_1 & \ddots & \vdots \\
\vdots & \ddots & \ddots &  a_2 \\
0 & \cdots & 0 &  a_1 \\
\end{bmatrix} \in R^{{\ell}\times {\ell}}.
\end{align*}
Let $I = \{i_1,\dots,i_s\}$ be an ordered set of row indices and let $J = \{j_1,\dots,j_s\}$ be an ordered set of column indices. The square $s \times s$ submatrix $A^I_J$ is called \textbf{proper} if for all $m \in \{1,\dots, \ell \}$ the inequality $i_m \leq j_m$ holds. This means that proper submatrices of $A$ are exactly those that have the chance that their determinant is nonzero.
\end{definition}

\begin{definition}
Let $A$ be a $\ell \times \ell$ upper triangular Toeplitz matrix
\begin{align*}
    A = \begin{bmatrix}
a_1 & a_2 & \cdots & a_{\ell} \\
0 & a_1 & \ddots & \vdots \\
\vdots & \ddots & \ddots &  a_2 \\
0 & \cdots & 0 &  a_1 \\
\end{bmatrix}
\end{align*}
where $a_i \in R$ for $i=1,2\dots, \ell$. $A$ is said to be \textbf{$\gamma$-superregular} if the determinant of every proper submatrix of $A$ is a unit in $R$.\\
$A$ is called \textbf{reverse $\gamma$-superregular} if $A$ is $\gamma$-superregular and also the reverse matrix
\begin{align*}
    A_{rev} = \begin{bmatrix}
a_{\ell} & a_{l-1} & \cdots & a_{1} \\
0 & a_{\ell} & \ddots & \vdots \\
\vdots & \ddots & \ddots &  a_{\ell-1} \\
0 & \cdots & 0 &  a_{\ell} \\
\end{bmatrix}
\end{align*}
is $\gamma$-superregular.
\end{definition}

Over finite field the definition of $\gamma$-superregularity is equivalent to saying that every proper submatrix of $A$ has nonzero determinant, which indeed coincides with the notion of superregular matrices.

Let $\pi : R \longrightarrow R/\mathfrak{m}$ be the canonical projection modulo the maximal ideal of $R$. We have the following easy result.

\begin{proposition} \label{superregularimpliespsuperregular}
$A\in R^{n\times n}$ is superregular if and only if $\pi(A)\in  R/\mathfrak{m}$ is superregular.
\end{proposition}
\begin{proof}
Let $\Tilde{A}$ be a proper submatrix of $A$ in $R$. Then  $\Tilde{A}$ is invertible in $R$ if and only if $\det(\Tilde{A})$ is a unit in $R$. This is equivalent to say that $\pi(\det(\Tilde{A}))\ne 0$ in $R/\mathfrak{m}$, but  $\pi(\det(\Tilde{A}))=\det(\pi(\Tilde{A}))$, which means that $\pi(\Tilde{A})$ is invertible in $R/\mathfrak{m}$.
\end{proof}

We extend a $\gamma$-superregular matrix so that its rows form a $\gamma$-generator sequence, because by Theorem \ref{pgenseq} the rows of a truncated sliding generator matrix of a convolutional code over $R$ always form a $\gamma$-generator sequence. When doing this, we need to make sure that condition (b) of \ref{general} is satisfied. Here, Proposition \ref{plinindep} comes into play since it gives a connection between $\gamma$-superregular matrices and $\gamma$-linearly independent rows.

 The following proposition shows how to obtain a (reverse) $\gamma$-superregular upper triangular block Toeplitz matrix from a (reverse) $\gamma$-superregular upper triangular Toeplitz matrix. In \cite{virtu2012}, this construction has been provided only over finite fields. Here, for the sake of completeness, we only state the following ring-version of this result.

\begin{proposition} \label{extractrows}
Let $A$ be an $\ell \times \ell$ (reverse) $\gamma$-superregular matrix defined over $R$ with $\ell := (L+1)(n+k-1)$. For $j = 0,1,\dots,L$ let $I_j$ and $J_j$ be defined as follows:
\begin{align*}
    I_j &:= \{jn + j(k-1)+1, jn + j(k-1)+2, \dots, (j+1)n + j(k-1)\} \\
    J_j &:= \{(j+1)n + j(k-1), (j+1)n + j(k-1)+1, \dots, (j+1)(n+k-1)\}
\end{align*}
and denote by $I$ and $J$ the union of these sets, i.e.
\begin{align*}
    I = \bigcup_{j=0}^L I_j, \; \; J = \bigcup_{j=0}^L J_j.
\end{align*}
Then $\Tilde{A} = A_I^J$ (and $\tilde{A}_{rev}$) are $(L+1)k \times (L+1)n$ upper block triangular submatrices of $A$ (and $A_{rev}$) such that every $(L+1)k \times (L+1)k$ full size minor of $\Tilde{A}$ (and $\tilde{A}_{rev}$) formed from the columns with indices $1 \leq t_1 < \dots < t_{(L+1)k}$, where $t_{sk+1} > sn$ for $s = 1, 2, \dots L$ is a unit in $R$.
\end{proposition}

We conclude this section by providing an example for the construction of an MDP and reverse MDP convolutional code via $\gamma$-superregular and reverse $\gamma$-superregular matrices.
In order to be able to use the preceding proposition as well as Theorem \ref{general}, we need to assume that $\nu|k$, $k|\delta$ and $\nu n-k > \delta$, i.e. $L=\frac{\delta}{k}=\deg(G(z))$.

\begin{example}
We want to construct a reverse MDP $(3,2,2)$-convolutional code over $\mathbb{Z}_{11^2}$, i.e. $n = 3, \; \; k = 2, \; \; \delta=2, \; \; p=\gamma=11, \; \;r=\nu=2$.

Note that $\nu|k$, $k|\delta$ and $\nu n-k > \delta$ and we have $L=\frac{\delta}{k} =1$.
We start by constructing a reverse $\ell \times \ell$ $\gamma$-superregular matrix over $\mathbb{Z}_{11^2}$ where $\ell = (L+1)(n+k-1) = 6$. Using a computer algebra program one can check that the following $6 \times 6$ matrix $T$ is reverse superregular over $\mathbb{Z}_{11}$:
\begin{align*}
T=\begin{bmatrix}
    1 & 2 & 1 & 1 & 3 & 4 \\
    0 & 1 & 2 & 1 & 1 & 3 \\
    0 & 0 & 1 & 2 & 1 & 1 \\
    0 & 0 & 0 & 1 & 2 & 1 \\
    0 & 0 & 0 & 0 & 1 & 2 \\
    0 & 0 & 0 & 0 & 0 & 1
\end{bmatrix}
\end{align*}
and therefore it is in particular reverse $\gamma$-superregular over $\mathbb{Z}_{11^2}$. Applying Theorem \ref{extractrows} to this matrix yields $I_0 = \{1, 2, 3\}, I_1 = \{4, 5, 6\}, J_0 = \{3\}, J_1= \{6\}$ and therefore 
\begin{align*}
    \Tilde{G}_L^c=
    \begin{bmatrix}
    1 & 2 & 1 & 1 & 3 & 4 \\
    0 & 0 & 0 & 1 & 2 & 1 
    \end{bmatrix}
    =
    \begin{bmatrix}
    \Tilde{G}_0 & \Tilde{G}_1 \\
     & \Tilde{G}_0
    \end{bmatrix}.
\end{align*}
As a next step we extend $\Tilde{G}_L^c$ to $G_L^c$.
\begin{align*}
    G_L^c=
    \begin{bmatrix}
    \Tilde{G}_0 & \Tilde{G}_1 \\
    11\cdot\Tilde{G}_0 & 11\cdot\Tilde{G}_1 \\
     & \Tilde{G}_0 \\
     & 11\cdot \Tilde{G}_0
    \end{bmatrix}
    =
    \begin{bmatrix}
    G_0 & G_1 \\
     & G_0
    \end{bmatrix}
    =
    \begin{bmatrix}
    1 & 2 & 1 & 1 & 3 & 4 \\
    11 & 22 & 11 & 11 & 33 & 44 \\
    0 & 0 & 0 & 1 & 2 & 1 \\
    0 & 0 & 0 & 11 & 22 & 11 
    \end{bmatrix}.
\end{align*}
By construction, 

\begin{align*}
    G(z) = G_0 + G_1z = 
    \begin{bmatrix}
        1 + z & 2 + 3z & 1 + 4z \\
        11+11z & 22 + 33z & 11 + 44z
    \end{bmatrix}
\end{align*}
is a $\gamma$-encoder of a reverse MDP $(3, 2, 2)$-convolutional code over $\mathbb{Z}_{11^2}$.
\end{example}

\bigskip

\section{Acknowledgements}
This work was partially supported by  Swiss National Science Foundation grant n. 188430. The authors are thankful to Alessandro Neri for fruitful comments and discussion and to the anonymous reviewers whose comments contributed to improve the paper.

\bigskip

\bibliographystyle{abbrv}

\bibliography{references}

\end{document}